\documentclass[11pt,a4paper]{article}
 
\usepackage[T1]{fontenc}
\usepackage[utf8]{inputenc}
\usepackage[english]{babel}
 
\usepackage{newtxtext}      
\usepackage{mathtools,amsthm}        
\usepackage{newtxmath}

\usepackage[margin=1in]{geometry}
\usepackage{microtype}
\usepackage{setspace}
\usepackage{enumitem,booktabs,array}
 
\usepackage{bm,mathrsfs}
\allowdisplaybreaks

\usepackage[dvipsnames]{xcolor}
\usepackage{hyperref}
\hypersetup{
  colorlinks=true,
  linkcolor=blue!50!black,
  citecolor=magenta!70!black,
  urlcolor=blue!70!black,
  pdftitle={A Two-HCIZ Gaussian Matrix Model for Non-Intersecting Brownian Bridges},
  pdfauthor={M. Kosmakov}
}

\usepackage{csquotes}
\usepackage[
  backend=biber,
  doi=true,
  isbn=false,
  url=false,
  style=numeric,
  giveninits=true
]{biblatex}
\renewbibmacro{in:}{}

\newtheorem{theorem}{Theorem}[section]

\newtheorem{proposition}[theorem]{Proposition}
\newtheorem{corollary}[theorem]{Corollary}
\theoremstyle{definition}\newtheorem{definition}[theorem]{Definition}
\theoremstyle{remark}\newtheorem{remark}[theorem]{Remark}

\numberwithin{equation}{section}

\DeclareMathOperator{\Tr}{Tr}

\DeclareMathOperator{\diag}{diag}
\newcommand{\R}{\mathbb{R}}
\newcommand{\C}{\mathbb{C}}
\newcommand{\e}{\mathrm{e}}
\newcommand{\dd}{\mathrm{d}}
\newcommand{\HH}{\mathbb{H}(n)}
\newcommand{\UU}{\mathrm{U}(n)}
\newcommand{\EE}{\mathbb{E}}

\usepackage{titling}
\pretitle{\begin{center}\LARGE\bfseries}
\posttitle{\par\vspace{0.4em}\end{center}}
\preauthor{\begin{center}\large\scshape}
\postauthor{\par\end{center}}
\predate{}\postdate{}

\makeatletter
\newcommand{\affilmark}[1]{\textsuperscript{#1}}
\newcommand{\affiltext}[2]{%
  \begingroup
    \def\@thefnmark{#1}%
    \@footnotetext{#2}%
  \endgroup
}
\makeatother

\begin{document}

\title{A Two-HCIZ Gaussian Matrix Model for Non-Intersecting Brownian Bridges}

\author{Maksim Kosmakov\affilmark{$\circ$}}
\date{}
\maketitle

\affiltext{$\circ$}{\textit{Department of Mathematical Sciences, University of Cincinnati,\\
P.O. Box 210025, Cincinnati, OH 45221, USA.} \quad \textit{E-mail:} \href{mailto:kosmakmm@ucmail.uc.edu}{kosmakmm@ucmail.uc.edu}}

\begin{abstract}
We construct a unitarily invariant Hermitian matrix ensemble whose fixed-time eigenvalue law coincides with the Karlin--McGregor law for non-intersecting Brownian bridges with arbitrary finite multiplicities at both endpoints. This provides an explicit matrix-ensemble realization of the known mixed-type multiple orthogonal polynomial and Riemann--Hilbert description of the general multi-start/multi-end problem. We then derive several exact finite-$n$ consequences of this construction. These include a path-space lift as an orbital Hermitian Brownian bridge and a reduction of the partition function to a single compact HCIZ integral with explicit $t$-dependence. We also compare the one-sided reduction with the Gaussian external-field ensemble, showing that, although the two ensembles are spectrally equivalent, their angular statistics are different. Finally, we derive fixed-time Schwinger--Dyson identities and associated resolvent relations for the dressed ensemble.
\end{abstract}

\noindent\textbf{Keywords:}   Non-intersecting Brownian bridges; Harish--Chandra--Itzykson--Zuber integral; random matrices; multiple orthogonal polynomials; Riemann--Hilbert problem; Toda hierarchy.

\section{Introduction}

Non-intersecting Brownian motions and Brownian bridges play a central role in random matrix theory, integrable probability, and statistical physics. They provide fundamental examples of strongly correlated stochastic systems and exhibit universal fluctuation phenomena.  In the statistical-physics literature they appear as systems of  vicious walkers, where the non-crossing constraint produces strong collective effects and leads to exactly solvable structures in a number of regimes~\cite{deGennes1968}. In random matrix theory, the canonical matrix-valued dynamical analogue is Dyson Brownian motion~\cite{Dyson1962}, namely the Hermitian matrix diffusion whose ordered eigenvalues evolve as a non-colliding process.

At the same time, the Harish--Chandra--Itzykson--Zuber (HCIZ) integral and its large-$n$ asymptotics play an important role in random matrix theory, representation theory, and statistical mechanics. There is by now a substantial literature connecting non-intersecting Brownian systems, HCIZ asymptotics, and the Matytsin variational picture in the large-$n$ regime~\cite{GuionnetMaida2005,BunBouchaudMajumdarPotters2014,GrelaMajumdarSchehr2021}.  A precise understanding of the HCIZ large-$n$ limit is also of direct interest in rotationally invariant inference problems, including extensive-rank matrix factorization and denoising~\cite{MaillardKrzakalaMezardZdeborova2022,TroianiErbaKrzakalaMaillardZdeborova2022}.

The central object of the present paper is the fixed-end analogue of non-colliding systems, namely non-intersecting Brownian bridges (NIBBs), that is, systems of one-dimensional Brownian bridges conditioned never to intersect. At a fixed intermediate time, their positions form a determinantal point process governed by the Karlin--McGregor formula~\cite{KarlinMcGregor1959}. From the matrix-model point of view, the fully degenerate one-start/one-end case reduces to the Gaussian Unitary Ensemble (GUE), while closely related one-start/multiple-end generalizations arise through Gaussian random matrices with external source~\cite{BrezinHikami1998}. In the present paper we consider the general finite-$n$ setting with several starting points and several ending points, allowing arbitrary finite multiplicities at both sides.

For such boundary data, the fixed-time distribution is naturally described by a mixed-type multiple orthogonal polynomial ensemble in the sense of Daems--Kuijlaars and by the associated Riemann--Hilbert problem (RHP)~\cite{DaemsKuijlaars2007,Kuijlaars2010}.  Within this framework, important asymptotic and structural results for related external-source and mixed-MOP settings were obtained in~\cite{Bleher2011,BleherDelvauxKuijlaars2011,AriznabarretaManas2015,BerezinKuijlaarsParra2023}. Thus, the general multi-start/multi-end fixed-time law was already understood at the level of multiple orthogonal polynomials and the associated Riemann--Hilbert problem, but not through a corresponding matrix ensemble. In the classical GUE and one-sided external-field settings, by contrast, one has the full correspondence
\begin{equation}
\text{Matrix ensemble}
\Longleftrightarrow
\text{Multiple orthogonal polynomials}
\Longleftrightarrow
\text{Riemann--Hilbert problem}.
\end{equation}
Motivated by this gap, we introduce a Hermitian matrix model in which a Gaussian weight is dressed by two HCIZ factors, one encoding the initial data and one encoding the terminal data. If $A$ and $B$ are diagonal matrices containing the starting and ending locations, and if $t\in(0,1)$ with $\sigma_t^2=t(1-t)$, we consider
\begin{equation*}
\dd\mu_{A,B,t}(M)\propto
\exp\bigl(-\frac{1}{2\sigma_t^2}\Tr(M^2)\bigr)
\int_{\UU}\exp\bigl(\frac{1}{t}\Tr(AUMU^\dagger)\bigr)\dd U\,
\int_{\UU}\exp\bigl(\frac{1}{1-t}\Tr(BVMV^\dagger)\bigr)\dd V\,\dd M.
\end{equation*}
This is the two-HCIZ Gaussian ensemble that we study in the paper.

The focus of the present paper is the exact theory at finite $n$ and the structures that already appear at that level. Much of the existing literature on HCIZ-type models and non-intersecting systems focuses instead on large-$n$ asymptotics, variational descriptions, and limiting spectral behavior. Here, by contrast, we introduce the matrix model underlying the general non-intersecting bridge law and study it from three complementary viewpoints: as a matrix realization of that law, as a source of exact structural identities, and as a unitarily invariant counterpart of the one-sided external-field ensemble. We expect these exact results to also provide useful guidance for the large-$n$ analysis.

The starting point is the finite-$n$ bridge problem. Theorem~\ref{thm:KM-equivalence} identifies the eigenvalue law of the two-HCIZ ensemble with the Karlin--McGregor distribution for non-intersecting Brownian bridges with prescribed boundary data. In this way, the model provides a matrix-integral realization of the general finite-$n$ bridge law. It also admits a natural path-space lift: Theorem~\ref{thm:orbital_bridge_lift} shows that, for each $t\in(0,1)$, the fixed-time ensemble arises as the time-$t$ marginal of an orbital Brownian bridge on Hermitian matrices whose endpoint laws are supported on the unitary orbits of $A$ and $B$.

The matrix representation has further structural consequences. Corollary~\ref{cor:collapse} shows that, after integrating out the Gaussian matrix variable, the partition function reduces to a single compact HCIZ integral times an explicit $t$-dependent prefactor. Equivalently, the nontrivial part of the partition function is identified with the compact HCIZ integral, viewed as a $2$D Toda $\tau$-function under the Miwa specialization determined by $A$ and $B$. At the level of the fixed-time ensemble, Proposition~\ref{prop:master_SD} gives orbital Ward identities and the associated Schwinger--Dyson equations. These, in turn, imply exact resolvent relations for spectral observables; see Corollary~\ref{cor:physical_resolvent_SD}.

A further point, relevant from both the matrix-model and statistical-mechanics perspectives, is that spectral equivalence does not determine the full matrix law. In the one-sided reduction $B=bI_n$, after a scalar shift the two-HCIZ model has the same eigenvalue distribution as the Gaussian external-field ensemble studied, for example, in~\cite{Bleher2011,BleherDelvauxKuijlaars2011}. The two matrix laws nevertheless differ: the external-field ensemble singles out the eigenbasis of the source matrix, whereas the two-HCIZ model remains unitarily invariant. The two models therefore agree on spectral observables but differ on angular observables, including eigenvector overlaps and related isotropy diagnostics~\cite{ PaccoRos2023,BenigniCipolloni2024}. From the statistical-mechanics viewpoint, this reflects two different observational regimes: the HCIZ formulation is natural when one averages over orientations, while the external-field formulation is appropriate when a distinguished laboratory basis breaks rotational symmetry, as in transport problems with fixed scattering channels~\cite{JKMS2023,Beenakker1997}.

This unitary invariance has a concrete consequence for expectations of matrix observables: once the angular and spectral variables are separated, conjugation-invariant quantities reduce to spectral data. In particular, the expectations of matrix powers are scalar matrices, determined by the corresponding trace moments and hence by spectral integrals against the one-point density, or equivalently against the mixed-MOP Christoffel--Darboux kernel; see Proposition~\ref{prop:spectral_kernel_resolvent}.

The paper is organized as follows. Section~\ref{sec:preliminaries} introduces the Brownian-bridge setting and recalls the Karlin--McGregor formula. Section~\ref{sec:spectral-realization} defines the two-HCIZ ensemble and develops its basic finite-$n$ realization, including the orbital bridge interpretation and the partition-function reduction. Section~\ref{sec:reductions} discusses special cases and compares the model with Gaussian external-field ensembles. Section~\ref{sec:unitary_moments} uses unitary invariance to represent moments and resolvents in terms of spectral data. Section~\ref{sec:integrable} treats the additional exact structure of the model, including the Toda interpretation and the fixed-time Ward and Schwinger--Dyson identities. Section~\ref{sec:discussion} closes with a discussion of the results and some directions for future work. The appendices contain auxiliary proofs.

\section*{Acknowledgments}

The author would like to thank A.~Prokhorov for useful discussions and P.~Bleher for sparking the author's interest in external-field models. The author is also grateful to the anonymous reviewers for their careful reading and constructive comments, which significantly improved the presentation of the paper.

\section{Preliminaries and the Karlin--McGregor Law}\label{sec:preliminaries}

We introduce the notation used throughout the paper and recall the Karlin--McGregor formula for the fixed-time positions of non-intersecting Brownian bridges.  

\subsection{Notation and setup}
\label{subsec:notation}

Throughout the paper, $t\in(0,1)$ denotes a fixed observation time, and we set $\sigma_t^2:=t(1-t)$. We write $\HH$ for the real Hilbert space of $n\times n$ Hermitian matrices,
  and $\UU$ for the unitary group. The symbols $\dd U$ and $\dd V$ always
denote normalized Haar probability measure on $\UU$. We write $f(x)\propto g(x)$ when $f(x)=C\,g(x)$ for some constant $C$ independent of $x$.

We denote by
\begin{equation*}
\mathcal W_n:=\{x=(x_1,\dots,x_n)\in\R^n:\ x_1<\cdots<x_n\},
\qquad
\overline{\mathcal W}_n:=\{x=(x_1,\dots,x_n)\in\R^n:\ x_1\le\cdots\le x_n\}
\end{equation*}
the open and closed Weyl chambers of type $A_{n-1}$.

The boundary data are specified by $p$ distinct starting points
$a_1<\cdots<a_p$ with multiplicities $m_1,\dots,m_p$, and
$q$ distinct ending points $b_1<\cdots<b_q$ with multiplicities
$n_1,\dots,n_q$. We encode them in the diagonal matrices
$A=\diag(\bm a)$ and $B=\diag(\bm b)$, where
\begin{equation}
\bm a=
(\underbrace{a_1,\dots,a_1}_{m_1},\dots,\underbrace{a_p,\dots,a_p}_{m_p})
\in\overline{\mathcal W}_n,
\qquad
\bm b=
(\underbrace{b_1,\dots,b_1}_{n_1},\dots,\underbrace{b_q,\dots,b_q}_{n_q})
\in\overline{\mathcal W}_n,
\end{equation}
with
\begin{equation}
a_1<\cdots<a_p,
\qquad
b_1<\cdots<b_q,
\qquad
\sum_{\ell=1}^p m_\ell=n,
\qquad
\sum_{k=1}^q n_k=n.
\end{equation}

For $M\in\HH$, we write its eigenvalues in weakly increasing order as
$\lambda_1\le \cdots\le \lambda_n$,
$\bm\lambda=(\lambda_1,\dots,\lambda_n)\in\overline{\mathcal W}_n$,
and $\Delta(\bm\lambda):=\prod_{1\le i<j\le n}(\lambda_j-\lambda_i)$ denotes the Vandermonde determinant.

We denote by
\begin{equation}
p_s(x,y)=\frac{1}{\sqrt{2\pi s}}
\exp\!\left(-\frac{(x-y)^2}{2s}\right),
\qquad s>0,
\end{equation}
the one-dimensional heat kernel, and by
\begin{equation}
p_s(X,Y)=\frac{1}{(2\pi s)^{n^2/2}}
\exp\!\left(-\frac{1}{2s}\Tr(X-Y)^2\right),
\qquad s>0,
\end{equation}
the heat kernel on $\HH$.

A system of $n$ non-intersecting Brownian bridges (NIBBs) from $\bm a$ to
$\bm b$ on $[0,1]$ is a collection
\begin{equation}
\bm X(t)=(X_1(t),\dots,X_n(t)),\qquad 0\le t\le 1,
\end{equation}
whose law is that of $n$ independent one-dimensional Brownian bridges
from $a_i$ to $b_i$, conditioned never to intersect. Thus $\bm X(0)=\bm a$, $\bm X(1)=\bm b$, and $\bm X(t)\in \mathcal W_n$ for $0<t<1$.

\subsection{Karlin--McGregor law and HCIZ formula}

For a single Brownian bridge from $a$ at time $0$ to $b$ at time $1$, the
intermediate-time density is given by 
\begin{equation}
\varrho_{a,b,t}(x)=\frac{p_t(a,x)\,p_{1-t}(x,b)}{p_1(a,b)}
\propto
\exp\!\left(\frac{ax}{t}+\frac{bx}{1-t}-\frac{x^2}{2\sigma_t^2}\right).
\end{equation}

For $n$ non-intersecting Brownian bridges with ordered boundary data
$\bm a,\bm b\in\overline{\mathcal W}_n$, the Karlin--McGregor formula gives the
fixed-time density on $\mathcal W_n$.

\begin{proposition}[Karlin--McGregor]
The joint density of the positions
$\bm\lambda=(\lambda_1,\dots,\lambda_n)\in \mathcal W_n$ of $n$ non-intersecting
Brownian bridges at time $t\in(0,1)$ is
\begin{equation}
\label{eq:KM-reduced}
\varrho_{\mathrm{KM}}(\bm\lambda;\bm a,\bm b,t)
\propto
\det\!\bigl[\e^{a_i\lambda_j/t}\bigr]_{i,j=1}^n\,
\det\!\bigl[\e^{b_k\lambda_j/(1-t)}\bigr]_{k,j=1}^n\,
\e^{-\frac{1}{2\sigma_t^2}\sum_{j=1}^n\lambda_j^2}.
\end{equation}
In the confluent case, repeated start or end points are
understood in the usual limiting sense.
\end{proposition}

We also recall the Harish--Chandra--Itzykson--Zuber \cite{HarishChandra1957,ItzyksonZuber1980} integral. For Hermitian
matrices $X,Y\in\HH$ with eigenvalues
$\bm x=(x_1,\dots,x_n)$ and $\bm y=(y_1,\dots,y_n)$, one has
\begin{equation}
\int_{\UU}\exp\!\big(\Tr(XUYU^\dagger)\big)\,\dd U
=
\left(\prod_{k=1}^{n-1}k!\right)
\frac{\det[\e^{x_i y_j}]_{i,j=1}^n}{\Delta(\bm x)\Delta(\bm y)}.
\end{equation}
If some eigenvalues coincide, the right-hand side is understood by confluent continuation.

\section{The Two-HCIZ Ensemble and its Spectral Realization}
\label{sec:spectral-realization}

\subsection{Spectral and path-space realization}
We now define our central object. Let $A$ and $B$ be the diagonal matrices encoding the start and end points as defined in Subsection~\ref{subsec:notation}.

\begin{definition} 
The two-HCIZ dressed Gaussian measure on $\HH$ is defined by
\begin{equation}
\label{eq:twoHCIZmeasure}
\dd\mu_{A,B,t}(M) = \frac{1}{Z_{A,B,t}} \e^{-\frac{1}{2\sigma_t^2}\Tr M^2} \left(\int_{\UU}\e^{\frac{1}{t}\Tr(AUMU^\dagger)}\dd U\right) \left(\int_{\UU}\e^{\frac{1}{1-t}\Tr(BVMV^\dagger)}\dd V\right) \dd M,
\end{equation}
where $\dd M$ denotes the Euclidean Lebesgue measure on $\HH$ induced by the inner product $\langle X,Y\rangle=\Tr(XY)$,  and $Z_{A,B,t}$ is the normalization constant (\textit{partition function})
\begin{equation}
 Z_{A,B,t}=\int_{\HH} \e^{-\frac{1}{2\sigma_t^2}\Tr M^2} \left(\int_{\UU}\e^{\frac{1}{t}\Tr(AUMU^\dagger)}\dd U\right) \left(\int_{\UU}\e^{\frac{1}{1-t}\Tr(BVMV^\dagger)}\dd V\right) \dd M.
\end{equation}
\end{definition}
This model consists of a standard GUE-type Gaussian measure on $M$, which is then dressed by two separate HCIZ-type integrals. The first integral couples the eigenvalues of $M$ to the start configuration $A$, while the second couples them to the end configuration $B$. Our first result shows that this ensemble gives exactly the fixed-time law of non-intersecting Brownian bridges.

\begin{theorem}[Finite-$n$ spectral equivalence]
\label{thm:KM-equivalence}
The joint eigenvalue density of the two-HCIZ ensemble \eqref{eq:twoHCIZmeasure} coincides with the Karlin--McGregor law \eqref{eq:KM-reduced} for any finite $n$, including the confluent case of repeated start or end points.
\end{theorem}

\begin{proof}
Diagonalize $M=W\Lambda W^\dagger$ with $\Lambda=\diag(\lambda_1,\dots,\lambda_n)$.  
By Weyl's integration formula (see, for example, \cite{Mehta2004}),
\begin{equation}
\dd M = \frac{(2\pi)^{n(n-1)/2}}{\prod_{k=1}^{n-1} k!} \Delta(\bm\lambda)^2 \prod_{j=1}^n \dd \lambda_j\,\dd W .
\end{equation}
Since the remaining weight depends on $M$ only through its eigenvalues, the integral over $W$ equals $1$.

Applying the HCIZ formula to the two orbital factors gives
\begin{equation*}
\int_{\UU} e^{\frac1t \Tr(AU\Lambda U^\dagger)}\,\dd U
= C_n(t)  \frac{\det[e^{a_i\lambda_j/t}]}{\Delta(\bm a)\Delta(\bm\lambda)},
\quad
\int_{\UU} e^{\frac1{1-t} \Tr(BV\Lambda V^\dagger)}\,\dd V
= C_n(1-t) \frac{\det[e^{b_i\lambda_j/(1-t)}]}{\Delta(\bm b)\Delta(\bm\lambda)},
\end{equation*}
with $C_n(s)=s^{n(n-1)/2}\prod_{k=1}^{n-1}k!$.  Substituting into \eqref{eq:twoHCIZmeasure}, one obtains
\begin{equation}
\varrho_{A,B,t}(\bm\lambda)
\propto
\Delta(\bm\lambda)^2
e^{-\frac1{2\sigma_t^2}\sum_{j=1}^n \lambda_j^2}
\frac{\det[e^{a_i\lambda_j/t}]}{\Delta(\bm a)\Delta(\bm\lambda)}
\frac{\det[e^{b_i\lambda_j/(1-t)}]}{\Delta(\bm b)\Delta(\bm\lambda)}.
\end{equation}
Since the factors $\Delta(\bm a)^{-1}$ and $\Delta(\bm b)^{-1}$ are independent of $\bm\lambda$, they can be    absorbed into the normalization. Therefore, for $\bm\lambda\in\mathcal W_n$,
\begin{equation}
\varrho_{A,B,t}(\bm\lambda)
\propto
\det[e^{a_i\lambda_j/t}]
\det[e^{b_i\lambda_j/(1-t)}]
e^{-\frac1{2\sigma_t^2}\sum_{j=1}^n \lambda_j^2}.
\end{equation}
 The confluent case is obtained by the standard HCIZ limiting procedure when eigenvalues of $A$ or $B$ coalesce.
\end{proof}

The two-HCIZ ensemble also admits a natural path-space realization. More precisely, it can be realized as the time-$t$ marginal of a Brownian bridge on Hermitian matrices whose endpoints are distributed on the conjugacy classes of $A$ and $B$ with the canonical heat-kernel coupling. This is the matrix-valued analogue of the usual Doob-transform / non-intersecting-bridge picture.

To formulate this path-space realization, we introduce the endpoint coupling
\begin{equation}
\label{eq:endpoint_coupling}
\Pi_{A,B}(\dd X\,\dd Y)
:=
\frac{1}{\mathcal Z_{A,B}}\,
p_1(X,Y)\,\nu_A(\dd X)\,\nu_B(\dd Y),
\end{equation}
where
\begin{equation}
\mathcal Z_{A,B}
:=
\int_{\mathcal O_A}\int_{\mathcal O_B}
p_1(X,Y)\,\nu_A(\dd X)\,\nu_B(\dd Y).
\end{equation}
Here
$\mathcal O_A:=\{UAU^\dagger:U\in\UU\}$ and
$\mathcal O_B:=\{VBV^\dagger:V\in\UU\}$
are the unitary orbits of $A$ and $B$, and
$\nu_A,\nu_B$ are the corresponding orbital Haar probability measures.

We write $C([0,1],\HH)$ for the space of continuous $\HH$-valued paths on $[0,1]$. If $M_\bullet\in C([0,1],\HH)$, then $M_s$ denotes its value at time $s$. For each $(X,Y)\in \mathcal O_A\times\mathcal O_B$, let $\mathbf P^{X,Y}$ denote the law on $C([0,1],\HH)$ of the Hermitian Brownian bridge from $X$ at time $0$ to $Y$ at time $1$. Mixing these bridge laws against the endpoint coupling $\Pi_{A,B}$, we obtain a probability measure $\mathbf Q_{A,B}$ on $C([0,1],\HH)$ given by
\begin{equation}
\mathbf Q_{A,B}(\dd M_\bullet)
:=
\int_{\mathcal O_A}\int_{\mathcal O_B}
\mathbf P^{X,Y}(\dd M_\bullet)\,
\Pi_{A,B}(\dd X\,\dd Y).
\end{equation}

\begin{theorem}[Orbital Brownian-bridge]
\label{thm:orbital_bridge_lift}
For $\mathbf Q_{A,B}(\dd M_\bullet)$ the following hold:
\begin{enumerate}
\item The endpoint pair has law $(M_0,M_1)\sim \Pi_{A,B}$,
and in particular $M_0\sim \nu_A$, $M_1\sim \nu_B$.

\item For every $t\in(0,1)$, the time-$t$ marginal of $\mathbf Q_{A,B}$ has
density
\begin{equation}
\label{eq:bridge_marginal_factorized}
q_t(M)\propto
\left(\int_{\mathcal O_A} p_t(X,M)\,\nu_A(\dd X)\right)
\left(\int_{\mathcal O_B} p_{1-t}(M,Y)\,\nu_B(\dd Y)\right)
\end{equation}
with respect to Lebesgue measure on $\HH$.

\item Equivalently, $  q_t(M) \,\dd M=
\dd\mu_{A,B,t}(M)$,
so the two-HCIZ ensemble is exactly the
time-$t$ marginal of the orbital Brownian bridge $\mathbf Q_{A,B}$.
\end{enumerate}
\end{theorem}

\begin{proof}
For a Brownian bridge from $X$ to $Y$, the time-$t$ density is
\begin{equation}
\frac{p_t(X,M)\,p_{1-t}(M,Y)}{p_1(X,Y)}.
\end{equation}
Integrating \eqref{eq:endpoint_coupling} against the conditional bridge law from
$X$ to $Y$, the factor $p_1(X,Y)$ cancels, and one obtains \eqref{eq:bridge_marginal_factorized}. The endpoint pair is distributed according to $\Pi_{A,B}$ by construction, and the marginal statements follow immediately. For example, the $X$-marginal is proportional to
\begin{equation}
\left(\int_{\mathcal O_B} p_1(X,Y)\,\nu_B(\dd Y)\right)\nu_A(\dd X),
\end{equation}
and the inner integral is constant on $\mathcal O_A$ by conjugation
invariance.

Finally, expanding $\Tr(X-M)^2=\Tr X^2+\Tr M^2-2\Tr(XM)$. Since $\Tr X^2=\Tr A^2$ on $\mathcal O_A$, one has
\begin{align*}
\int_{\mathcal O_A} p_t(X,M)\,\nu_A(\dd X)
&=
C_t(A)\,
\exp\!\left(-\frac{1}{2t}\Tr M^2\right)
\int_{\UU}\exp\!\left(\frac{1}{t}\Tr(AUMU^\dagger)\right)\dd U,
\end{align*}
with $C_t(A)$ independent of $M$, and likewise
\begin{equation*}
\int_{\mathcal O_B} p_{1-t}(M,Y)\,\nu_B(\dd Y)
=
C_{1-t}(B)\,
\exp\!\left(-\frac{1}{2(1-t)}\Tr M^2\right)
\int_{\UU}\exp\!\left(\frac{1}{1-t}\Tr(BVMV^\dagger)\right)\dd V.
\end{equation*}
Multiplying the two factors in \eqref{eq:bridge_marginal_factorized} gives
\begin{equation}
q_t(M)\propto
\exp\!\left(-\frac{1}{2\sigma_t^2}\Tr M^2\right)
\int_{\UU}\exp\!\left(\frac1t\Tr(AUMU^\dagger)\right)\dd U
\int_{\UU}\exp\!\left(\frac1{1-t}\Tr(BVMV^\dagger)\right)\dd V,
\end{equation}
thus \(q_t(M)\,\dd M=\dd\mu_{A,B,t}(M)\), proving \textup{(iii)}.
\end{proof}
\begin{remark}
\label{rem:Doob_bridge}
The process $\mathbf Q_{A,B}$ is the natural Brownian bridge on $\HH$
associated with the orbital endpoint data $(\nu_A,\nu_B)$. Equivalently, it may
be viewed as the Doob $h$-transform of Hermitian Brownian motion by
\begin{equation}
h_s(M):=\int_{\mathcal O_B} p_{1-s}(M,Y)\,\nu_B(\dd Y),
\end{equation}
started from the initial law $\nu_A$.
\end{remark}

Projecting the orbital bridge onto eigenvalues recovers the standard
Doob-transform description of non-intersecting Brownian bridges in the Weyl
chamber. Thus the matrix-valued construction above provides a finite-$n$ matrix
derivation of the bridge dynamics familiar from the non-intersecting Brownian
motion / Dyson-bridge literature; compare, for example,
\cite{Dyson1962,KatoriTanemura2003,GrelaMajumdarSchehr2021}.

\begin{corollary}
\label{cor:orbital_bridge_eigenvalue_dynamics}
Let $\mathbf Q_{A,B}$ be the orbital Brownian bridge from
Theorem~\ref{thm:orbital_bridge_lift}, and let
$\bm\lambda(s)$, $0<s<1$, denote the ordered eigenvalues of $M_s$.
For every $0<s<1$, one has $\bm\lambda(s)\in\mathcal W_n$, and on $(0,1)$ the
process is a time-inhomogeneous diffusion with generator
\begin{equation}
\label{eq:generator_Doob_K_main}
(\mathcal G_s f)(\bm\lambda)
=
\frac12\sum_{i=1}^n \partial_i^2 f(\bm\lambda)
+
\sum_{i=1}^n
\bigl(\partial_i \log K_s^{(B)}(\bm\lambda)\bigr)\,\partial_i f(\bm\lambda),
\end{equation}
where $K_s^{(B)}$ denotes the positive terminal Karlin--McGregor $h$-function
on $\mathcal W_n$ associated with the endpoint data $B$.  If $b_1<\cdots<b_n$, then
\begin{equation}
\label{eq:K_s_terminal_main}
K_s^{(B)}(\bm\lambda)
:=
\det\!\bigl[p_{1-s}(\lambda_i,b_j)\bigr]_{i,j=1}^n,
\end{equation}
whereas in the presence of repeated terminal points, $K_s^{(B)}$ is understood
in the standard confluent Karlin--McGregor sense. Equivalently,
\begin{equation}
\dd \lambda_i(s)
=
\dd X_i(s)
+
\partial_i\log K_s^{(B)}(\bm\lambda(s))\,\dd s,
\qquad i=1,\dots,n,
\end{equation}
where $X_1,\dots,X_n$ are independent standard Brownian motions.
If $b_1<\cdots<b_n$, this may be written explicitly as
\begin{equation}
\dd \lambda_i(s)
=
\dd X_i(s)
+
\left(
-\frac{\lambda_i(s)}{1-s}
+
\partial_i\log\det[\exp(b_k\lambda_j(s)/(1-s))]_{k,j=1}^n
\right)\dd s.
\end{equation}
\end{corollary}
\begin{proof}
    The proof is given in Appendix~\ref{app:orbital_bridge_generator}.
\end{proof}
\begin{remark}
   Note that the forward generator depends only on the terminal data $B$; the
initial data $A$ enter through the initial law at time $s=0$.  
\end{remark} 

\subsection{Partition function}

We next consider the partition function. Completing the square reduces it to a single compact HCIZ integral.

\begin{corollary}[Single-HCIZ collapse]
\label{cor:collapse}
The partition function $Z_{A,B,t}$ admits the closed form
\begin{equation}\label{eq:singleHCIZ}
Z_{A,B,t}
=(2\pi\sigma_t^2)^{\frac{n^2}{2}}\,
\exp\!\Big(\frac{1-t}{2t}\Tr A^2+\frac{t}{2(1-t)}\Tr B^2\Big)
\int_{\UU}\exp\!\big(\Tr(AWBW^\dagger)\big)\,\dd W .
\end{equation}
\end{corollary}

\begin{proof}
For fixed $U,V\in\UU$, set
\begin{equation}
\mu(U,V):=(1-t)U^\dagger A U+t\,V^\dagger B V.
\end{equation}
Then
\begin{equation}
-\frac{1}{2\sigma_t^2}\Tr M^2
+\Tr\!\Big(\Big(\frac1t\,U^\dagger A U+\frac1{1-t}\,V^\dagger B V\Big)M\Big)
=
-\frac{1}{2\sigma_t^2}\Tr\!\big(M-\mu(U,V)\big)^2
+\frac{1}{2\sigma_t^2}\Tr\mu(U,V)^2.
\end{equation}
Thus, conditionally on $(U,V)$,
\begin{equation}
    M=X+\mu(U,V),
\end{equation}
where $X$ is a centered Hermitian Gaussian random matrix with law proportional to $\exp(-\Tr X^2/(2\sigma_t^2))\,dX$.
Therefore the $M$-integral equals
$(2\pi\sigma_t^2)^{\frac{n^2}{2}}
\exp\!\Big(\frac{1}{2\sigma_t^2}\Tr\mu(U,V)^2\Big)$,
with
\begin{equation}
\Tr\mu(U,V)^2
=
(1-t)^2\Tr A^2+t^2\Tr B^2
+2t(1-t)\Tr\!\big(U^\dagger A U\,V^\dagger B V\big).
\end{equation}
Setting $W:=VU^\dagger$, we get
$\Tr\!\big(U^\dagger A U\,V^\dagger B V\big)=\Tr\!\big(AW^\dagger B W\big)$, 
and by Haar bi-invariance the change of variables $V\mapsto W=VU^\dagger$ preserves Haar measure, so the integrand becomes independent of $U$.   Therefore,
\begin{equation}
Z_{A,B,t}
=
(2\pi\sigma_t^2)^{n^2/2}
\exp\!\Big(\frac{1-t}{2t}\Tr A^2+\frac{t}{2(1-t)}\Tr B^2\Big)
\int_{\UU}\exp\!\big(\Tr(AW^\dagger BW)\big)\,\dd W.
\end{equation}
Finally, by Haar invariance under \(W\mapsto W^\dagger\), this equals
\eqref{eq:singleHCIZ}.
\end{proof}

Among the simplest genuinely two-sided families are the rank-$(r,\ell)$
two-point spectrum model and the balanced signature model. In these cases, the single-HCIZ collapse reduces the partition function to Jacobi-type integrals on $(0,1)$ with a linear exponential tilt, connecting the present finite-$n$ model to families studied in~\cite{EdelmanSutton2008,BasorChenEhrhardt2010,ForresterWitte2002}.

For two nonnegative integers $r$ and $\ell$ satisfying $r+\ell\le n$, let
$p=\min(r,\ell)$ and define
\begin{equation} 
Z_{n,r,\ell}:=\int_{(0,1)^p} \Delta(\lambda)^2 \prod_{i=1}^{p} \lambda_i^{|r-\ell|}(1-\lambda_i)^{\,n-r-\ell}\,\dd \lambda.
\end{equation}
\begin{proposition}
\label{prop:tilted_jacobi}
For the following two families, the two-HCIZ partition function reduces to a tilted
Jacobi unitary ensemble on $(0,1)$.

\begin{enumerate}
\item[\textup{(i)}] \textup{Projection / two-point spectrum case.}
Let
\begin{equation}
A=\diag(\underbrace{a,\dots,a}_{r},\underbrace{0,\dots,0}_{n-r}),
\qquad
B=\diag(\underbrace{b,\dots,b}_{\ell},\underbrace{0,\dots,0}_{n-\ell}),
\end{equation}
and define
\begin{equation}
p:=\min(r,\ell),\qquad
j:=n-r-\ell,\qquad
k:=|r-\ell|.
\end{equation}
Then the partition function is proportional to a Jacobi-type integral on $(0,1)$
with exponents $(k,j)$ and linear tilt parameter $\gamma=ab$, namely
\begin{equation}\label{eq:proj_direct_jacobi}
Z_{A,B,t}
=
K^{\mathrm{proj}}_{n,t,a,b;r,\ell}
\int_{(0,1)^p}
\Delta(\lambda)^2
\prod_{i=1}^{p}
\lambda_i^{k}(1-\lambda_i)^j e^{ab\lambda_i}\,\dd \lambda,
\end{equation}
where
\begin{equation}\label{eq:proj_direct_jacobi_const}
K^{\mathrm{proj}}_{n,t,a,b;r,\ell}
=
(2\pi\sigma_t^2)^{n^2/2}
\exp\!\Big(\frac{1-t}{2t}a^2r+\frac{t}{2(1-t)}b^2\ell\Big)\,
Z_{n,r,\ell}^{-1}.
\end{equation}

\item[\textup{(ii)}] \textup{Balanced signature case.}
Assume $n=2m$, and let
\begin{equation}
J:=\diag(I_m,-I_m),\qquad A=aJ,\qquad B=bJ.
\end{equation}
Then the partition function is proportional to
\begin{equation}\label{eq:sig_direct_jacobi}
Z_{A,B,t}
=
K^{\mathrm{sig}}_{n,t,a,b}
\int_{(0,1)^m}
\Delta(\lambda)^2
\prod_{i=1}^{m} e^{4ab\,\lambda_i}\,\dd \lambda,
\end{equation}
where
\begin{equation}\label{eq:sig_direct_jacobi_const}
K^{\mathrm{sig}}_{n,t,a,b}
=
(2\pi\sigma_t^2)^{n^2/2}
\exp\!\Big(\frac{1-t}{2t}na^2+\frac{t}{2(1-t)}nb^2\Big)\,
e^{-2mab}\,
Z_{2m,m,m}^{-1}.
\end{equation}
\end{enumerate}
\end{proposition}

\begin{proof}
The proof is given in Appendix~\ref{app:tilted_jacobi}.
\end{proof}

\begin{remark}
The projection formula above is stated in the regime $r+\ell\le n$, where the upper-left block of a Haar unitary has the standard Jacobi block density \eqref{eq:appendix_jacobi_block_law}. In the complementary regime $r+\ell>n$, the block has $r+\ell-n$ deterministic singular values equal to $1$, while the remaining nontrivial singular values are obtained from the complementary $(n-r)\times (n-\ell)$ block and thus again satisfy a Jacobi law after the change of variables $\lambda\mapsto 1-\lambda$. 
\end{remark}

\begin{remark}
Although the finite-$n$ reductions above are explicit, the corresponding
large-$n$ asymptotics of the compact HCIZ factors require a more delicate
analysis, see \cite{GuionnetMaida2005} for related results on HCIZ asymptotics. The corresponding
matrix-valued measures for these models also require separate study.
\end{remark}

\section{Special Reductions and Spectral versus Angular Observables}
\label{sec:reductions}

In this section we discuss two special reductions of the two-HCIZ ensemble at the matrix level. The scalar case recovers the centered GUE. More importantly, the one-sided reduction identifies the two-HCIZ ensemble as a unitarily invariant counterpart of the Gaussian external-field model: the two ensembles have the same eigenvalue law, but different matrix laws. This comparison will serve as a guide for the later analysis of spectral and angular observables.

\subsection{Special reductions}

\paragraph{The scalar reduction}

If $A=aI_n$ and $B=bI_n$, then after the scalar translation
$\widetilde M=M-((1-t)a+tb)I_n$ the two-HCIZ ensemble reduces to a centered GUE
with variance $\sigma_t^2=t(1-t)$.

\paragraph{The one-sided reduction and the external-field ensemble}

If $B=bI_n$, then after the scalar shift $\widetilde M:=M-tbI_n$ the two-HCIZ ensemble takes the form
\begin{equation}
\dd\mu_{A,bI,t}(\widetilde M)\;\propto\;
\e^{-\frac{1}{2\sigma_t^2}\Tr \widetilde M^{\,2}}
\left(\int_{\UU}\e^{\frac{1}{t}\Tr(AU\widetilde M U^\dagger)}\,\dd U\right)\dd \widetilde M.
\end{equation}
Its eigenvalue law coincides with that of the Gaussian external-field ensemble
\begin{equation}
\dd\nu_{A,t}(M)
\;\propto\;
\exp\!\Big(-\frac{1}{2\sigma_t^2}\Tr M^2+\frac{1}{t}\Tr(AM)\Big)\,\dd M.
\end{equation}

The two matrix laws are nevertheless different. In the external-field ensemble, the source matrix $A$ selects a preferred basis, so the law is not conjugation invariant and angular observables retain information about the orientation of the eigenvectors relative to that basis. By contrast, the two-HCIZ ensemble remains unitarily invariant: conditional on the spectrum, the eigenbasis is Haar distributed. The two ensembles therefore agree on spectral observables, but differ on angular observables such as eigenvector overlaps and related isotropy diagnostics~\cite{PaccoRos2023,BenigniCipolloni2024}.

From the statistical-mechanics perspective, these correspond to different
observational regimes. The HCIZ formulation is natural when one averages over orientations and focuses on basis-independent observables, while the
external-field formulation is appropriate when a distinguished laboratory basis breaks rotational symmetry, for example through prescribed leads or transport channels~\cite{JKMS2023,Beenakker1997}. This distinction motivates the next sections: unitary invariance reduces matrix moments and resolvents to spectral quantities, while the Haar eigenbasis controls the angular statistics.

\subsection{Angular statistics}

By construction the two-HCIZ density \eqref{eq:twoHCIZmeasure} is invariant
under unitary conjugation. Passing to spectral coordinates
\begin{equation}
M=\Psi\Lambda\Psi^\dagger,
\qquad
\Lambda=\diag(\lambda_1,\dots,\lambda_n),
\end{equation}
the two-HCIZ weight depends only on the spectral variable $\Lambda$, whereas
the angular variable $\Psi$ contributes only through the Haar factor in Weyl's
integration formula. It follows that the conditional law of $\Psi$ at fixed
$\Lambda$ is Haar. Therefore for every integrable observable $\mathcal O_{\mathrm{ang}}(\Psi)$,
\begin{equation}
\EE \!\left[\mathcal O_{\mathrm{ang}}(\Psi)\mid \Lambda\right]
=
\EE_{\mathrm{Haar}}\!\left[\mathcal O_{\mathrm{ang}}\right].
\end{equation}
 Consequently, for every integrable spectral observable $\mathcal O_{\mathrm{spec}}(\Lambda)$ and every integrable angular observable $\mathcal O_{\mathrm{ang}}(\Psi)$,
\begin{equation}
\EE \big[\mathcal O_{\mathrm{spec}}(\Lambda)\,\mathcal O_{\mathrm{ang}}(\Psi)\big]
=
\EE_{\mathrm{KM}}\big[\mathcal O_{\mathrm{spec}}\big]\,
\EE_{\mathrm{Haar}}\big[\mathcal O_{\mathrm{ang}}\big],
\end{equation}
where $\EE_{\mathrm{KM}}$ denotes expectation with respect to the
Karlin--McGregor spectral law \eqref{eq:KM-reduced}. Accordingly, all standard overlap statistics are universal and coincide with
their classical Haar values; see, for example,
\cite{Mehta2004}, \cite{AGZ2010}.

\begin{corollary} 
Fix an
eigenvector $\psi_j$ of $M$. Let $u\in\C^n$ be a deterministic unit vector, let
$e_1,\dots,e_n$ be a deterministic orthonormal basis, and let $P$ be a
deterministic orthogonal projector of rank $r$. Then the following hold.

\begin{enumerate}
\item The one-dimensional overlap satisfies $|\langle u,\psi_j\rangle|^2 \sim \mathrm{Beta}(1,n-1)$.

\item The projector overlap satisfies $\langle \psi_j,P\psi_j\rangle \sim \mathrm{Beta}(r,n-r)$.

\item The coordinate masses
$\bigl(|\langle e_1,\psi_j\rangle|^2,\dots,|\langle e_n,\psi_j\rangle|^2\bigr) 
\sim \mathrm{Dirichlet}(1,\dots,1)$.

\item For the inverse participation ratios $ Y_q(\psi_j):=\sum_{m=1}^n |\langle e_m,\psi_j\rangle|^{2q}$,
  $ q\in\mathbb N$, one has
\begin{equation}
\EE\big[Y_q(\psi_j)\big]
=
\frac{n!\,q!}{(n+q-1)!}.
\end{equation}
\end{enumerate}
\end{corollary}

\begin{remark} 
More generally, for products of entries of $\Psi$ and their complex conjugates one
has the unitary Weingarten expansion
\begin{equation}
\EE_{\mathrm{Haar}}\!\left[\prod_{a=1}^{m}\Psi_{i_a j_a}\,\overline{\Psi}_{i_a' j_a'}\right]
=
\sum_{\sigma,\tau\in S_m}
\prod_{a=1}^{m}\delta_{i_a,i'_{\sigma(a)}}
\prod_{a=1}^{m}\delta_{j_a,j'_{\tau(a)}}
\,\mathrm{Wg}_n(\tau\sigma^{-1}),
\end{equation}
where $\mathrm{Wg}_n$ is the unitary Weingarten function; see the recent discussion in~\cite{CollinsMatsumotoNovak2022} and the references therein.
\end{remark}

\section{Moment Reduction and Spectral Representations}
\label{sec:unitary_moments}

\subsection{Moments as spectral integrals and resolvents}

Using unitary invariance, matrix observables reduce naturally to spectral ones. We begin with the simplest example, namely the trace, whose distribution can be computed explicitly.

\begin{proposition}
Let
\begin{equation} 
X:=\Tr M,
\qquad
\mu_t:=(1-t)\Tr A+t\,\Tr B.
\end{equation}
Then the moment generating function of $X$ is
\begin{equation}\label{eq:trace_mgf_shift_formula}
\mathcal M_{A,B,t}(s):=\EE\big[e^{s\Tr M}\big]
= \exp\!\left(\mu_t s+\frac{n\sigma_t^2}{2}s^2\right).
\end{equation}
Thus $\Tr M$ is exactly Gaussian with mean $\mu_t$ and variance $n\sigma_t^2$.
In particular, for every $r\ge 0$,
\begin{equation} 
\EE\big[(\Tr M)^r\big]
=
\left.\frac{\partial^r}{\partial s^r}
\exp\!\left(\mu_t s+\frac{n\sigma_t^2}{2}s^2\right)\right|_{s=0}
=
r!\sum_{m=0}^{\lfloor r/2\rfloor}
\frac{\mu_t^{\,r-2m}}{(r-2m)!\,m!}
\left(\frac{n\sigma_t^2}{2}\right)^m.
\end{equation}
\end{proposition}

\begin{proof}
Shifting $A\mapsto A+t s I_n$ in \eqref{eq:twoHCIZmeasure} produces the factor $e^{s\Tr M}$, thus $\mathcal M_{A,B,t}(s)=\frac{Z_{A+t s I_n,B,t}}{Z_{A,B,t}}$.  Now apply the single-collapse formula \eqref{eq:singleHCIZ}. Since
\begin{equation}
\Tr(A+t s I_n)^2=\Tr A^2+2ts\,\Tr A+n t^2 s^2,
\end{equation}
one obtains
\begin{equation}
Z_{A+t s I_n,B,t}
=
\exp\!\Big((1-t)s\Tr A+t s\Tr B+\frac{n t(1-t)}{2}s^2\Big)\,Z_{A,B,t},
\end{equation}
which is exactly \eqref{eq:trace_mgf_shift_formula}.
\end{proof}

\begin{remark}
The completed-square representation also gives access to exact finite-$n$ polynomial moments. Beyond the purely Gaussian contribution, however, these moments involve HCIZ-tilted angular averages rather than plain Haar averages, thus combining Gaussian Wick contractions with unitary angular integrals.
\end{remark}

\begin{proposition}
\label{prop:spectral_kernel_resolvent}
Let  $\varrho^{\mathrm{sym}}_{A,B,t}$ denote the symmetrization of  the normalized Karlin--McGregor joint density  $\varrho_\mathrm{KM}$,   define the one-point density $\rho^{(1)}_{A,B,t}$ by
\begin{equation} 
\rho^{(1)}_{A,B,t}(x)
=
n\int_{\R^{n-1}}
\varrho^{\mathrm{sym}}_{A,B,t}(x,\lambda_2,\dots,\lambda_n)\,
\dd\lambda_2\cdots \dd\lambda_n,
\qquad
\int_{\R}\rho^{(1)}_{A,B,t}(x)\,\dd x=n.
\end{equation}
Then:
\begin{enumerate}
\item[\textup{(i)}] For every $k\ge0$,
\begin{equation}\label{eq:trace_moment_one_point_density}
\EE[M^k]
=
\frac1n\,\EE[\Tr(M^k)]\,I_n,
\qquad
\EE[\Tr(M^k)]
=
\int_{\R} x^k \rho^{(1)}_{A,B,t}(x)\,\dd x.
\end{equation}

\item[\textup{(ii)}] Since the symmetrized eigenvalue law is a determinantal mixed-type multiple orthogonal polynomial ensemble in the sense of~\cite{DaemsKuijlaars2007,Kuijlaars2010}, its one-point density is given by
\begin{equation}
\rho^{(1)}_{A,B,t}(x)=K_n(x,x),
\end{equation}
where $K_n$ is the mixed-MOP Christoffel--Darboux kernel. Thus
\begin{equation} 
\EE[\Tr(M^k)]
=
\int_{\R} x^k K_n(x,x)\,\dd x.
\end{equation}

\item[\textup{(iii)}] The spectral resolvent
$\Omega(z):=\EE[\Tr(z-M)^{-1}]$, $z\in\C\setminus\R$,
is the Stieltjes transform of the one-point density:
\begin{equation}\label{eq:resolvent_one_point}
\Omega(z)
=
\int_{\R}\frac{\rho^{(1)}_{A,B,t}(x)}{z-x}\,\dd x
=
\int_{\R}\frac{K_n(x,x)}{z-x}\,\dd x.
\end{equation}
\end{enumerate}
\end{proposition}

\begin{proof}
Diagonalizing $M=\Psi \diag(\lambda_1,\dots,\lambda_n)\Psi^\dagger$, we have $\Tr(M^k)=\sum_{j=1}^n \lambda_j^k$. By the definition of the one-point density, for every polynomial $g$,
\begin{equation}
\EE\!\left[\sum_{j=1}^n g(\lambda_j)\right]
=
\int_{\R} g(x)\,\rho^{(1)}_{A,B,t}(x)\,\dd x.
\end{equation}
Taking $g(x)=x^k$ gives
\begin{equation}
\EE[\Tr(M^k)]
=
\int_{\R} x^k \rho^{(1)}_{A,B,t}(x)\,\dd x.
\end{equation}

Since the matrix law is unitarily invariant, $\EE[M^k]$ commutes with every unitary matrix and hence must be a scalar multiple of the identity. Taking traces gives \eqref{eq:trace_moment_one_point_density}.

Item \textup{(ii)} is the standard determinantal one-point correlation formula for the mixed-type multiple orthogonal polynomial ensemble associated with the symmetrized Karlin--McGregor law. Finally,
\begin{equation}
\Tr(z-M)^{-1}=\sum_{j=1}^n \frac1{z-\lambda_j},
\qquad z\in\C\setminus\R,
\end{equation}
so taking expectation and using the definition of $\rho^{(1)}_{A,B,t}$ gives \eqref{eq:resolvent_one_point}.
\end{proof}

\section{Integrable structure of the partition function and fixed-time identities}
\label{sec:integrable}

Section~\ref{sec:unitary_moments} used only unitary invariance and the resulting spectral description. We now turn to the additional structure specific to the two-HCIZ weight itself. At the level of the partition function, the single-HCIZ collapse identifies the nontrivial factor as the compact HCIZ integral, naturally viewed as a Miwa-specialized $2$D Toda $\tau$-function. At the level of the fixed-time ensemble, the same structure yields orbital Ward identities, Schwinger--Dyson equations, and the resulting exact resolvent relations.

\subsection{Miwa specialization and $2$D Toda structure}

It is well known that partition functions of determinantal ensembles are moment determinants, and hence satisfy Desnanot--Jacobi identities and the Toda equations. Accordingly, the two-HCIZ partition function gives rise to a $2$D Toda $\tau$-function. In the present model this $\tau$-function admits an explicit identification: after the single-HCIZ collapse, the nontrivial factor of the partition function is precisely the compact HCIZ integral evaluated at the Miwa data determined by $A$ and $B$.

\begin{proposition}[Miwa specialization and $2$D Toda structure]
Set
\begin{equation} 
\mathsf t_m^{(+)}:=\frac1m\Tr A^m,
\qquad
\mathsf t_m^{(-)}:=\frac1m\Tr B^m,
\qquad m\ge1.
\end{equation}
Define
\begin{equation} 
\Phi_t\big(\mathsf t_2^{(+)},\mathsf t_2^{(-)}\big)
:=
\frac{1-t}{t}\,\mathsf t_2^{(+)}
+
\frac{t}{1-t}\,\mathsf t_2^{(-)}.
\end{equation}
Then
\begin{equation}\label{eq:Z_miwa_factorization}
Z_{A,B,t}
=
(2\pi t(1-t))^{\frac{n^2}{2}}
\exp\!\Big(\Phi_t\big(\mathsf t_2^{(+)},\mathsf t_2^{(-)}\big)\Big)\,
\tau_n(A,B),
\end{equation}
where
\begin{equation}\label{eq:HCIZ_tau_miwa}
\tau_n(A,B)=\tau_n\!\big(\{\mathsf t_m^{(+)}\},\{\mathsf t_m^{(-)}\}\big):=\int_{\UU}\exp\!\big(\Tr(AWBW^\dagger)\big)\,\dd W=\Big(\prod_{k=1}^{n-1}k!\Big)\,
\frac{\det[\e^{a_i b_j}]_{i,j=1}^n}{\Delta(\bm a)\,\Delta(\bm b)}.
\end{equation}
In particular, $\tau_n$, viewed as a function of the Miwa variables associated with $A$ and $B$, is a $2$D Toda $\tau$-function.
\end{proposition}

\begin{proof}
Equation \eqref{eq:Z_miwa_factorization} is simply Corollary~\ref{cor:collapse} rewritten in Miwa variables.
The determinant formula \eqref{eq:HCIZ_tau_miwa} is the classical HCIZ formula.
Since $\tau_n$ is a determinant of moment type, the standard bilinear determinant identities imply that it is a $2$D Toda $\tau$-function;  see, for example,
\cite{ZinnJustin2002,HarnadOrlov2003,Orlov2002}.
\end{proof}

\begin{remark}
The factorization \eqref{eq:Z_miwa_factorization} connects the model with the broader integrable-systems framework surrounding HCIZ/Toda-type matrix integrals, including character expansions, vertex-operator symmetries, and related $W$-algebra or Virasoro-type structures studied in the literature
\cite{Carrell2015,ZinnJustin2002,AdlerShiotaVanMoerbeke1998,AdlerShiotaVanMoerbeke1994,AdlerShiotaVanMoerbeke1995,AdlerVanMoerbeke1995}. In a related combinatorial direction, the same HCIZ factor is also connected with monotone Hurwitz theory and cut-and-join type formalisms
\cite{GouldenGuayPaquetNovak2014}. These related perspectives lie beyond the scope of the present paper, but they indicate natural directions for further study.
\end{remark}

 \subsection{Fixed-time orbital and Schwinger--Dyson identities}

We now consider the orbital Ward--Schwinger--Dyson identities for the fixed-time two-HCIZ ensemble, derived from the fixed-time matrix density
\begin{equation} 
q_t(M)\,\dd M
=
\frac{1}{Z_{A,B,t}}
\exp\!\Big(-\frac{1}{2\sigma_t^2}\Tr M^2\Big)\,
I_A(M)\,I_B(M)\,\dd M,
\qquad
\sigma_t^2=t(1-t),
\end{equation}
where  
\begin{equation} 
I_A(M):=\int_{\UU}\exp\!\Big(\frac1t\Tr(AUMU^\dagger)\Big)\,\dd U,
\qquad
I_B(M):=\int_{\UU}\exp\!\Big(\frac1{1-t}\Tr(MVBV^\dagger)\Big)\,\dd V.
\end{equation}

\begin{proposition}[Fixed-time Schwinger--Dyson identities]
\label{prop:master_SD}
\leavevmode
\begin{enumerate}
\item[\textup{(i)}] For every smooth compactly supported test function
$f:\HH\to\C$, one has
\begin{equation}\label{eq:master_SD_physical}
\EE\!\left[
\Delta f(M)
+
\Tr\!\Big(
\nabla f(M)
\Big[
-\frac{M}{\sigma_t^2}
+\frac{\mathbb A(M)}{t}
+\frac{\mathbb B(M)}{1-t}
\Big]
\Big)
\right]
=0,
\end{equation}
where the Hermitian matrix-valued source fields $\mathbb A(M)$ and
$\mathbb B(M)$ are defined by
\begin{equation}\label{eq:source_fields_defs_SD}
\frac{d}{d\varepsilon}\Big|_{\varepsilon=0}\log I_A(M+\varepsilon H)
=
\frac1t\Tr(\mathbb A(M)H),
\quad
\frac{d}{d\varepsilon}\Big|_{\varepsilon=0}\log I_B(M+\varepsilon H)
=
\frac1{1-t}\Tr(\mathbb B(M)H),
\end{equation}
for every $H\in\HH$.

\item[\textup{(ii)}] Let $\Lambda=\diag(\lambda_1,\dots,\lambda_n)$, and let
$\phi=(\phi_1,\dots,\phi_n)$ be a smooth compactly supported test vector field on
$\R^n$ satisfying
\begin{equation}
\phi_{\sigma(i)}(\lambda_{\sigma(1)},\dots,\lambda_{\sigma(n)})
=
\phi_i(\lambda_1,\dots,\lambda_n)
\end{equation}
for every permutation $\sigma\in S_n$ and every $i=1,\dots,n$. Then
\begin{equation}\label{eq:eigenvalue_SD_physical}
0=
\EE\!\left[
\sum_{i=1}^n \partial_{\lambda_i}\phi_i
+
2\sum_{1\le i<j\le n}\frac{\phi_i-\phi_j}{\lambda_i-\lambda_j}
-
\sum_{i=1}^n
\phi_i
\Big(
\frac{\lambda_i}{\sigma_t^2}
-\frac{\alpha_i(\Lambda)}{t}
-\frac{\beta_i(\Lambda)}{1-t}
\Big)
\right],
\end{equation}
where
\begin{equation} 
\alpha_i(\Lambda)=t\,\partial_{\lambda_i}\log I_A(\Lambda),
\qquad
\beta_i(\Lambda)=(1-t)\,\partial_{\lambda_i}\log I_B(\Lambda).
\end{equation}
\end{enumerate}
\end{proposition}

\begin{proof}
  Observe that $I_A(M)$ and
$I_B(M)$ are smooth positive functions on $\HH$ and their logarithmic directional derivatives are represented, via the trace pairing, by unique Hermitian matrices $\mathbb A(M)/t$ and $\mathbb B(M)/(1-t)$, which gives \eqref{eq:source_fields_defs_SD}. Integrating by parts against the vector field $\nabla f(M)\,q_t(M)$ then gives \eqref{eq:master_SD_physical}.

For part \textup{(ii)}, by the Weyl integration formula, the
 ordered eigenvalues  have  density
\begin{equation}
\rho(\lambda)
\propto
\Delta(\lambda)^2
\exp\!\Big(-\frac1{2\sigma_t^2}\sum_{i=1}^n \lambda_i^2\Big)
I_A(\Lambda)\,I_B(\Lambda).
\end{equation}
Since $\phi$ is smooth and compactly supported, and since $\rho$ vanishes
quadratically on the boundary where eigenvalues collide,
integration by parts on $\mathcal W_n$ gives
\begin{equation}
0
=
\int_{\mathcal W_n}
\sum_{i=1}^n
\partial_{\lambda_i}\!\big(\phi_i(\lambda)\rho(\lambda)\big)\,\dd\lambda=\EE\!\left[
\sum_{i=1}^n \partial_{\lambda_i}\phi_i
+
\sum_{i=1}^n \phi_i\,\partial_{\lambda_i}\log\rho(\lambda)
\right],
\end{equation}
with
\begin{equation}
\partial_{\lambda_i}\log\rho(\lambda)
=
2\sum_{j\ne i}\frac1{\lambda_i-\lambda_j}
-\frac{\lambda_i}{\sigma_t^2}
+\partial_{\lambda_i}\log I_A(\Lambda)
+\partial_{\lambda_i}\log I_B(\Lambda).
\end{equation}
Substituting this into the previous identity and using the definition of $\alpha_i(\Lambda)$ and $\beta_i(\Lambda)$ we get
\begin{equation}0= \EE\!\left[ \sum_{i=1}^n \partial_{\lambda_i}\phi_i+2\sum_{i=1}^n \phi_i\sum_{j\ne i}\frac1{\lambda_i-\lambda_j}-\sum_{i=1}^n \phi_i\Big( \frac{\lambda_i}{\sigma_t^2} -\frac{\alpha_i(\Lambda)}{t} -\frac{\beta_i(\Lambda)}{1-t} \Big) \right].
\end{equation}
Finally, using
\begin{equation}
\sum_{i=1}^n \phi_i\sum_{j\ne i}\frac1{\lambda_i-\lambda_j}
=
\sum_{1\le i<j\le n}\frac{\phi_i-\phi_j}{\lambda_i-\lambda_j},
\end{equation}
we obtain \eqref{eq:eigenvalue_SD_physical}.
\end{proof}

Let $\chi_R\in C_c^\infty(\R)$ satisfy $0\le \chi_R\le 1$ and $\chi_R\equiv 1$ on $[-R,R]$, and set
\begin{equation}
\phi_i^{(R)}(\Lambda):=\frac{\chi_R(\lambda_i)}{z-\lambda_i},
\qquad z\in\C\setminus\R.
\end{equation}
Applying Proposition~\ref{prop:master_SD}\textup{(ii)} to $\phi^{(R)}$ and then letting $R\to\infty$, using the Gaussian decay of the eigenvalue density, gives the fixed-time resolvent identity below.

\begin{corollary}[Fixed-time resolvent identity]
\label{cor:physical_resolvent_SD}
Define the spectral transforms
\begin{equation} 
\Omega(z):=\EE\!\left[\sum_{i=1}^n\frac{1}{z-\lambda_i}\right],
\qquad
\Omega_c(z,z):=
\EE\!\left[\left(\sum_{i=1}^n\frac{1}{z-\lambda_i}\right)^2\right]
-\Omega(z)^2,
\end{equation}
and
\begin{equation} 
H_A(z):=\EE\!\left[\sum_{i=1}^n\frac{\alpha_i(\Lambda)}{z-\lambda_i}\right],
\qquad
H_B(z):=\EE\!\left[\sum_{i=1}^n\frac{\beta_i(\Lambda)}{z-\lambda_i}\right].
\end{equation}
Then
\begin{equation}\label{eq:physical_resolvent_eq}
\Omega(z)^2+\Omega_c(z,z)
=
\frac{z}{\sigma_t^2}\,\Omega(z)
-\frac{n}{\sigma_t^2}
-\frac1t\,H_A(z)
-\frac1{1-t}\,H_B(z).
\end{equation}
\end{corollary}

\begin{remark}
In the general $(p,q)$ case \eqref{eq:physical_resolvent_eq} is not closed, since it couples $\Omega$ to the dressed transforms $H_A$ and $H_B$. We do not develop the corresponding closure hierarchy here.
\end{remark}

\section{Discussion}
\label{sec:discussion}

The primary focus of this work has been the exact finite-$n$ structure of the model, but the results also point toward several asymptotic directions. On the Brownian-bridge side, the two-HCIZ ensemble provides a matrix formulation from which one may seek a Matytsin-type variational description of the evolving empirical density. On the matrix-model side, the compact HCIZ reduction, together with the Ward and Schwinger--Dyson identities and the associated resolvent relations, suggests an approach based on loop equations, spectral curves, and nonlinear steepest descent for the underlying mixed-type Riemann--Hilbert problem.

The partition-function factorization also places the model within the broader integrable setting of Toda-type matrix integrals and HCIZ-type $\tau$-functions. It would be interesting to understand more fully how this perspective relates to Virasoro-type and related symmetry structures, including those appearing in \cite{AdlerVanMoerbekeVanderstichelen2012} for gap probabilities, and whether the present matrix realization helps clarify their finite-$n$ origin. In a different direction, the same HCIZ factor is also connected with monotone Hurwitz theory and cut-and-join formalisms~\cite{GouldenGuayPaquetNovak2014}.

Other natural directions include higher moments, higher-order resolvent identities, non-Gaussian deformations, other symmetry classes, and multitime observables.

\appendix
\section*{Appendices}

\section{Proof of Corollary~\ref{cor:orbital_bridge_eigenvalue_dynamics}}
\label{app:orbital_bridge_generator}

In this appendix we verify that the eigenvalue process of the orbital bridge
recovers the standard Doob-transform dynamics of non-intersecting Brownian
bridges in the Weyl chamber.

By Remark~\ref{rem:Doob_bridge}, the matrix-valued orbital bridge has generator
\begin{equation}
\mathcal L_s\Phi
=
\frac12\Delta_{\HH}\Phi
+
\langle \nabla\log h_s,\nabla\Phi\rangle,
\qquad
h_s(M):=\int_{\mathcal O_B} p_{1-s}(M,Y)\,\nu_B(\dd Y).
\end{equation}
Since $h_s$ is conjugation invariant, there exists a symmetric function
$\widehat h_s$ on $\overline{\mathcal W}_n$ such that
$h_s(M)=\widehat h_s(\bm\lambda(M))$. Therefore, for a conjugation-invariant test function
$\Phi(M)=f(\bm\lambda(M))$, one may pass to the radial part. The radial part
of $\frac12\Delta_{\HH}$ on conjugation-invariant functions is
\begin{equation}
\frac12\sum_{i=1}^n \partial_i^2
+
\sum_{i\neq j}\frac{1}{\lambda_i-\lambda_j}\partial_i.
\end{equation}
Moreover,
\begin{equation}
\langle \nabla\log h_s,\nabla\Phi\rangle
=
\sum_{i=1}^n
\bigl(\partial_i\log\widehat h_s(\bm\lambda)\bigr)\,\partial_i f(\bm\lambda).
\end{equation}
Thus the eigenvalue process has generator
\begin{equation}
\label{eq:generator_radial_bridge_raw}
(\mathcal G_s f)(\bm\lambda)
=
\frac12\sum_{i=1}^n \partial_i^2 f(\bm\lambda)
+
\sum_{i\neq j}\frac{1}{\lambda_i-\lambda_j}\partial_i f(\bm\lambda)
+
\sum_{i=1}^n
\bigl(\partial_i\log\widehat h_s(\bm\lambda)\bigr)\,\partial_i f(\bm\lambda).
\end{equation}

Assume now that $b_1<\cdots<b_n$. By the same orbital heat-kernel computation
as in the proof of Theorem~\ref{thm:orbital_bridge_lift},
\begin{equation}
\widehat h_s(\bm\lambda)
\propto
\exp\!\left(-\frac{1}{2(1-s)}\sum_{i=1}^n\lambda_i^2\right)
\int_{\UU}
\exp\!\left(\frac{1}{1-s}\Tr(BV\Lambda V^\dagger)\right)\,\dd V,
\qquad
\Lambda=\diag(\lambda_1,\dots,\lambda_n).
\end{equation}
Applying the Harish--Chandra--Itzykson--Zuber formula to the orbital integral gives
\begin{equation}
\widehat h_s(\bm\lambda)
=
C_{n,s,B}\,
\exp\!\left(-\frac{1}{2(1-s)}\sum_{i=1}^n\lambda_i^2\right)
\frac{\det[\exp(b_k\lambda_j/(1-s))]_{k,j=1}^n}{\Delta(\bm\lambda)},
\end{equation}
where $C_{n,s,B}$ is independent of $\bm\lambda$. Taking logarithmic
derivatives, we obtain
\begin{equation}
\partial_i\log\widehat h_s(\bm\lambda)
=
-\frac{\lambda_i}{1-s}
+
\partial_i\log\det[\exp(b_k\lambda_j/(1-s))]_{k,j=1}^n
-
\partial_i\log\Delta(\bm\lambda).
\end{equation}
Since
$\partial_i\log\Delta(\bm\lambda)
=
\sum_{j\neq i}\frac{1}{\lambda_i-\lambda_j}$,
the Vandermonde derivative cancels exactly the Dyson repulsion term in
\eqref{eq:generator_radial_bridge_raw}. Therefore
\begin{equation}
\label{eq:generator_after_cancellation}
(\mathcal G_s f)(\bm\lambda)
=
\frac12\sum_{i=1}^n \partial_i^2 f(\bm\lambda)
+
\sum_{i=1}^n
\left(
-\frac{\lambda_i}{1-s}
+
\partial_i\log\det[\exp(b_k\lambda_j/(1-s))]_{k,j=1}^n
\right)\partial_i f(\bm\lambda).
\end{equation}

Finally, using the one-dimensional heat kernel identity
\begin{equation}
p_{1-s}(x,b)
=
\frac{1}{\sqrt{2\pi(1-s)}}
\exp\!\left(
-\frac{x^2}{2(1-s)}
+\frac{bx}{1-s}
-\frac{b^2}{2(1-s)}
\right),
\end{equation}
we see that $K_s^{(B)}(\bm\lambda)$ from
\eqref{eq:K_s_terminal_main} satisfies
\begin{equation}
K_s^{(B)}(\bm\lambda)
=
\widetilde C_{n,s,B}\,
\exp\!\left(-\frac{1}{2(1-s)}\sum_{i=1}^n\lambda_i^2\right)
\det[\exp(b_k\lambda_j/(1-s))]_{k,j=1}^n,
\end{equation}
with $\widetilde C_{n,s,B}$ independent of $\bm\lambda$. Thus
\begin{equation}
\partial_i\log K_s^{(B)}(\bm\lambda)
=
-\frac{\lambda_i}{1-s}
+
\partial_i\log\det[\exp(b_k\lambda_j/(1-s))]_{k,j=1}^n.
\end{equation}
Substituting this into \eqref{eq:generator_after_cancellation} gives exactly
\eqref{eq:generator_Doob_K_main}.

The confluent case of repeated terminal points follows by the standard Karlin--McGregor/HCIZ coalescence procedure.

\section{Proof of Proposition~\ref{prop:tilted_jacobi}}
\label{app:tilted_jacobi}

We use the standard fact (see, for example, \cite{EdelmanSutton2008}) that if
$Y$ is the $r\times \ell$ upper-left block of a Haar unitary matrix in $\UU$,
with $r+\ell\le n$, and if $\lambda_1,\dots,\lambda_p\in(0,1)$,
$p=\min(r,\ell)$, are the nonzero eigenvalues of $YY^\dagger$, then their joint
density is
\begin{equation}\label{eq:appendix_jacobi_block_law}
\frac{1}{Z_{n,r,\ell}}
\Delta(\lambda)^2
\prod_{i=1}^{p}\lambda_i^{|r-\ell|}(1-\lambda_i)^{\,n-r-\ell}\,\dd \lambda.
\end{equation}

\begin{proof} 
For the projection case, assume $r+\ell\le n$. Then
\begin{equation}
A=\diag(\underbrace{a,\dots,a}_{r},\underbrace{0,\dots,0}_{n-r}),
\qquad
B=\diag(\underbrace{b,\dots,b}_{\ell},\underbrace{0,\dots,0}_{n-\ell}),
\end{equation}
so that $A=aP_r$ and $B=bP_\ell$, where $P_r$ and $P_\ell$ are the
diagonal projections of ranks $r$ and $\ell$. By
Corollary~\ref{cor:collapse},
\begin{equation}
Z_{A,B,t}
=
(2\pi\sigma_t^2)^{n^2/2}
\exp\!\Big(\frac{1-t}{2t}a^2r+\frac{t}{2(1-t)}b^2\ell\Big)
\int_{\UU} e^{\Tr(AWBW^\dagger)}\,\dd W.
\end{equation}
If $Y$ denotes the $r\times \ell$ upper-left block of $W$, then
\begin{equation}
\Tr(AWBW^\dagger)
=
ab\,\Tr(P_rWP_\ell W^\dagger)
=
ab\,\Tr(YY^\dagger)
=
ab\sum_{i=1}^p \lambda_i,
\end{equation}
where $\lambda_1,\dots,\lambda_p$ are the nonzero eigenvalues of $YY^\dagger$.
Using \eqref{eq:appendix_jacobi_block_law}, we obtain
\begin{equation}
\int_{\UU} e^{\Tr(AWBW^\dagger)}\,\dd W
=
Z_{n,r,\ell}^{-1}
\int_{(0,1)^p}
\Delta(\lambda)^2
\prod_{i=1}^{p}
\lambda_i^{|r-\ell|}(1-\lambda_i)^{\,n-r-\ell}
e^{ab\lambda_i}\,\dd \lambda,
\end{equation}
which is exactly \eqref{eq:proj_direct_jacobi}--\eqref{eq:proj_direct_jacobi_const}.

For the balanced signature case, let $n=2m$ and $J=2P-I$,
 where $P$ is the rank-$m$ projection onto the first $m$ coordinates. Then
\begin{equation}
\Tr(JWJW^\dagger)=4\Tr(PWPW^\dagger)-2m.
\end{equation}
Thus, for $A=aJ$ and $B=bJ$,
\begin{equation}
\Tr(AWBW^\dagger)
=
ab\,\Tr(JWJW^\dagger)
=
4ab\,\Tr(PWPW^\dagger)-2mab.
\end{equation}
Therefore
\begin{equation}
\int_{\mathrm U(2m)} e^{\Tr(AWBW^\dagger)}\,\dd W
=
e^{-2mab}
\int_{\mathrm U(2m)} e^{4ab\,\Tr(PWPW^\dagger)}\,\dd W.
\end{equation}
Applying the projection case with $r=\ell=m$, for which the Jacobi exponents
vanish, gives
\begin{equation}
\int_{\mathrm U(2m)} e^{\Tr(AWBW^\dagger)}\,\dd W
=
e^{-2mab} Z_{2m,m,m}^{-1}
\int_{(0,1)^m}
\Delta(\lambda)^2
\prod_{i=1}^{m} e^{4ab\lambda_i}\,\dd \lambda.
\end{equation}
Substituting this into Corollary~\ref{cor:collapse} gives
\eqref{eq:sig_direct_jacobi}--\eqref{eq:sig_direct_jacobi_const}.
\end{proof}

\printbibliography

@book{AGZ2010,
  AUTHOR = {Anderson, Greg W. and Guionnet, Alice and Zeitouni, Ofer},
  TITLE = {An Introduction to Random Matrices},
  SERIES = {Cambridge Studies in Advanced Mathematics},
  VOLUME = {118},
  PUBLISHER = {Cambridge University Press},
  ADDRESS = {Cambridge},
  YEAR = {2010},
  PAGES = {xiv+492},
  ISBN = {978-0-521-19452-5},
  MRCLASS = {60B20 (46L53 46L54)},
  MRNUMBER = {2760897},
  MRREVIEWER = {Terence\ Tao}
}

@article{AdlerShiotaVanMoerbeke1994,
  AUTHOR = {Adler, M. and Shiota, T. and van Moerbeke, P.},
  TITLE = {From the {$w_\infty$}-algebra to its central extension: a {$\tau$}-function approach},
  JOURNAL = {Phys. Lett. A},
  FJOURNAL = {Physics Letters. A},
  VOLUME = {194},
  YEAR = {1994},
  NUMBER = {1-2},
  PAGES = {33--43},
  ISSN = {0375-9601,1873-2429},
  MRCLASS = {58F07 (35Q58 58F37)},
  MRNUMBER = {1299512},
  MRREVIEWER = {Walter\ Oevel},
  DOI = {10.1016/0375-9601(94)00306-A},
  URL = {https://doi.org/10.1016/0375-9601(94)00306-A}
}

@article{AdlerShiotaVanMoerbeke1995,
  AUTHOR = {Adler, M. and Shiota, T. and van Moerbeke, P.},
  TITLE = {Random matrices, vertex operators and the {V}irasoro algebra},
  JOURNAL = {Phys. Lett. A},
  FJOURNAL = {Physics Letters. A},
  VOLUME = {208},
  YEAR = {1995},
  NUMBER = {1-2},
  PAGES = {67--78},
  ISSN = {0375-9601,1873-2429},
  MRCLASS = {82B41 (58F07 60B99 81R10 81T40)},
  MRNUMBER = {1359268},
  MRREVIEWER = {Oleksiy\ Khorunzhiy},
  DOI = {10.1016/0375-9601(95)00725-I},
  URL = {https://doi.org/10.1016/0375-9601(95)00725-I}
}

@incollection{AdlerVanMoerbeke1995,
  AUTHOR = {Adler, M. and van Moerbeke, P.},
  TITLE = {Matrix integrals, {T}oda symmetries, {V}irasoro constraints and orthogonal polynomials},
  BOOKTITLE = {R.{C}.{P}. 25, {V}ol. 47 ({S}trasbourg, 1993--1995)},
  SERIES = {Pr\'epubl. Inst. Rech. Math. Av.},
  VOLUME = {1995/24},
  PAGES = {215--262},
  PUBLISHER = {Univ. Louis Pasteur},
  ADDRESS = {Strasbourg},
  YEAR = {1995},
  MRCLASS = {58F07 (17B68 33C80 81R10)},
  MRNUMBER = {1461314}
}

@article{AriznabarretaManas2015,
  AUTHOR = {Ariznabarreta, Gerardo and Ma{\~n}as, Manuel},
  TITLE = {A Jacobi type Christoffel--Darboux formula for multiple orthogonal polynomials of mixed type},
  JOURNAL = {Linear Algebra Appl.},
  VOLUME = {468},
  YEAR = {2015},
  PAGES = {154--170},
  NOTE = {18th ILAS Conference},
  ISSN = {0024-3795},
  DOI = {10.1016/j.laa.2014.04.030},
  URL = {https://www.sciencedirect.com/science/article/pii/S0024379514002560}
}

@article{BasorChenEhrhardt2010,
  AUTHOR = {Basor, Estelle and Chen, Yang and Ehrhardt, Torsten},
  TITLE = {Painlev\'e {V} and time-dependent {J}acobi polynomials},
  JOURNAL = {J. Phys. A},
  FJOURNAL = {Journal of Physics. A. Mathematical and Theoretical},
  VOLUME = {43},
  YEAR = {2010},
  NUMBER = {1},
  PAGES = {015204, 25},
  ISSN = {1751-8113,1751-8121},
  MRCLASS = {33E17 (33C45 34M55)},
  MRNUMBER = {2570057},
  MRREVIEWER = {Andrei\ A.\ Kapaev},
  DOI = {10.1088/1751-8113/43/1/015204},
  URL = {https://doi.org/10.1088/1751-8113/43/1/015204}
}

@article{Beenakker1997,
  AUTHOR = {Beenakker, C. W. J.},
  TITLE = {Random-matrix theory of quantum transport},
  JOURNAL = {Rev. Mod. Phys.},
  VOLUME = {69},
  NUMBER = {3},
  PAGES = {731--808},
  YEAR = {1997},
  MONTH = jul,
  DOI = {10.1103/RevModPhys.69.731},
  URL = {https://link.aps.org/doi/10.1103/RevModPhys.69.731}
}

@article{BenigniCipolloni2024,
  AUTHOR = {Benigni, Lucas and Cipolloni, Giorgio},
  TITLE = {Fluctuations of eigenvector overlaps and the {B}erry conjecture for {W}igner matrices},
  JOURNAL = {Electron. J. Probab.},
  FJOURNAL = {Electronic Journal of Probability},
  VOLUME = {29},
  YEAR = {2024},
  PAGES = {Paper No. 150, 19},
  ISSN = {1083-6489},
  MRCLASS = {60B20 (15B52)},
  MRNUMBER = {4816009},
  MRREVIEWER = {Mostafa\ Sabri},
  DOI = {10.1214/24-EJP1203},
  URL = {https://doi.org/10.1214/24-EJP1203}
}

@article{BerezinKuijlaarsParra2023,
  AUTHOR = {Berezin, Sergey and Kuijlaars, Arno B. J. and Parra, Iv\'an},
  TITLE = {Planar orthogonal polynomials as type {I} multiple orthogonal polynomials},
  JOURNAL = {SIGMA Symmetry Integrability Geom. Methods Appl.},
  FJOURNAL = {SIGMA. Symmetry, Integrability and Geometry. Methods and Applications},
  VOLUME = {19},
  YEAR = {2023},
  PAGES = {Paper No. 020, 18},
  ISSN = {1815-0659},
  MRCLASS = {42C05 (30E25 41A21)},
  MRNUMBER = {4574010},
  MRREVIEWER = {Alfredo\ Dea\~no},
  DOI = {10.3842/SIGMA.2023.020},
  URL = {https://doi.org/10.3842/SIGMA.2023.020}
}

@incollection{Bleher2011,
  AUTHOR = {Bleher, Pavel M.},
  TITLE = {Lectures on Random Matrix Models: The Riemann--Hilbert Approach},
  BOOKTITLE = {Random Matrices, Random Processes and Integrable Systems},
  SERIES = {CRM Series in Mathematical Physics},
  EDITOR = {Harnad, John},
  PUBLISHER = {Springer},
  ADDRESS = {New York},
  PAGES = {251--349},
  YEAR = {2011},
  DOI = {10.1007/978-1-4419-9514-8_6},
  NOTE = {Preprint arXiv:0801.1858}
}

@article{BleherDelvauxKuijlaars2011,
  AUTHOR = {Bleher, P. and Delvaux, S. and Kuijlaars, A. B. J.},
  TITLE = {Random matrix model with external source and a constrained vector equilibrium problem},
  JOURNAL = {Comm. Pure Appl. Math.},
  FJOURNAL = {Communications on Pure and Applied Mathematics},
  VOLUME = {64},
  YEAR = {2011},
  NUMBER = {1},
  PAGES = {116--160},
  ISSN = {0010-3640,1097-0312},
  MRCLASS = {60B20 (30E25 82B44)},
  MRNUMBER = {2743878},
  MRREVIEWER = {Florent\ Benaych-Georges},
  DOI = {10.1002/cpa.20339},
  URL = {https://doi.org/10.1002/cpa.20339}
}

@article{BrezinHikami1998,
  AUTHOR = {Br{\'e}zin, E. and Hikami, S.},
  TITLE = {Level spacing of random matrices in an external source},
  JOURNAL = {Phys. Rev. E},
  FJOURNAL = {Physical Review E. Statistical, Nonlinear, and Soft Matter Physics},
  VOLUME = {58},
  YEAR = {1998},
  NUMBER = {6},
  PAGES = {7176--7185},
  ISSN = {1539-3755,1550-2376},
  MRCLASS = {82B41 (15A52)},
  MRNUMBER = {1662382},
  MRREVIEWER = {Oleksiy\ Khorunzhiy},
  DOI = {10.1103/PhysRevE.58.7176},
  URL = {https://doi.org/10.1103/PhysRevE.58.7176}
}

@article{BunBouchaudMajumdarPotters2014,
  AUTHOR = {Bun, J. and Bouchaud, J.-P. and Majumdar, S. N. and Potters, M.},
  TITLE = {Instanton Approach to Large $N$ Harish-Chandra-Itzykson-Zuber Integrals},
  JOURNAL = {Phys. Rev. Lett.},
  VOLUME = {113},
  NUMBER = {7},
  PAGES = {070201},
  YEAR = {2014},
  MONTH = aug,
  DOI = {10.1103/PhysRevLett.113.070201},
  URL = {https://link.aps.org/doi/10.1103/PhysRevLett.113.070201}
}

@article{Carrell2015,
  AUTHOR = {Carrell, S. R.},
  TITLE = {Diagonal solutions to the 2-{T}oda hierarchy},
  JOURNAL = {Math. Res. Lett.},
  FJOURNAL = {Mathematical Research Letters},
  VOLUME = {22},
  YEAR = {2015},
  NUMBER = {2},
  PAGES = {439--465},
  ISSN = {1073-2780,1945-001X},
  MRCLASS = {35Q53 (05E05 35C10 37K10)},
  MRNUMBER = {3342241},
  DOI = {10.4310/MRL.2015.v22.n2.a6},
  URL = {https://doi.org/10.4310/MRL.2015.v22.n2.a6}
}

@article{CollinsMatsumotoNovak2022,
  AUTHOR = {Collins, Beno{\^i}t and Matsumoto, Sho and Novak, Jonathan},
  TITLE = {The {W}eingarten calculus},
  JOURNAL = {Notices Amer. Math. Soc.},
  FJOURNAL = {Notices of the American Mathematical Society},
  VOLUME = {69},
  YEAR = {2022},
  NUMBER = {5},
  PAGES = {734--745},
  ISSN = {0002-9920,1088-9477},
  MRCLASS = {22C05 (05E16 22E30 60B20)},
  MRNUMBER = {4415894},
  MRREVIEWER = {Vladislav\ Kargin},
  DOI = {10.1090/noti2474},
  URL = {https://doi.org/10.1090/noti2474}
}

@article{DaemsKuijlaars2007,
  AUTHOR = {Daems, E. and Kuijlaars, A. B. J.},
  TITLE = {Multiple orthogonal polynomials of mixed type and non-intersecting {B}rownian motions},
  JOURNAL = {J. Approx. Theory},
  FJOURNAL = {Journal of Approximation Theory},
  VOLUME = {146},
  YEAR = {2007},
  NUMBER = {1},
  PAGES = {91--114},
  ISSN = {0021-9045,1096-0430},
  MRCLASS = {60J65 (30E25 42C05 82B41)},
  MRNUMBER = {2327475},
  MRREVIEWER = {Romain\ Abraham},
  DOI = {10.1016/j.jat.2006.12.001},
  URL = {https://doi.org/10.1016/j.jat.2006.12.001}
}

@article{deGennes1968,
  AUTHOR = {de Gennes, P.-G.},
  TITLE = {Soluble Model for Fibrous Structures with Steric Constraints},
  JOURNAL = {J. Chem. Phys.},
  VOLUME = {48},
  NUMBER = {5},
  PAGES = {2257--2259},
  YEAR = {1968},
  MONTH = mar,
  ISSN = {0021-9606},
  DOI = {10.1063/1.1669420},
  URL = {https://doi.org/10.1063/1.1669420}
}

@article{Dyson1962,
  AUTHOR = {Dyson, Freeman J.},
  TITLE = {A Brownian-Motion Model for the Eigenvalues of a Random Matrix},
  JOURNAL = {J. Math. Phys.},
  VOLUME = {3},
  NUMBER = {6},
  PAGES = {1191--1198},
  YEAR = {1962},
  MONTH = nov,
  ISSN = {0022-2488},
  DOI = {10.1063/1.1703862},
  URL = {https://doi.org/10.1063/1.1703862}
}

@article{EdelmanSutton2008,
  AUTHOR = {Edelman, Alan and Sutton, Brian D.},
  TITLE = {The beta-{J}acobi matrix model, the {CS} decomposition, and generalized singular value problems},
  JOURNAL = {Found. Comput. Math.},
  FJOURNAL = {Foundations of Computational Mathematics},
  VOLUME = {8},
  YEAR = {2008},
  NUMBER = {2},
  PAGES = {259--285},
  ISSN = {1615-3375,1615-3383},
  MRCLASS = {82B41 (33C45 65F15)},
  MRNUMBER = {2407033},
  MRREVIEWER = {Florent\ Benaych-Georges},
  DOI = {10.1007/s10208-006-0215-9},
  URL = {https://doi.org/10.1007/s10208-006-0215-9}
}

@article{ForresterWitte2002,
  AUTHOR = {Forrester, P. J. and Witte, N. S.},
  TITLE = {Application of the {$\tau$}-function theory of {P}ainlev\'e equations to random matrices: {$\mathrm{P}_{\mathrm{VI}}$}, the {JUE}, {CyUE}, {cJUE} and scaled limits},
  JOURNAL = {Nagoya Math. J.},
  FJOURNAL = {Nagoya Mathematical Journal},
  VOLUME = {174},
  YEAR = {2004},
  PAGES = {29--114},
  ISSN = {0027-7630,2152-6842},
  MRCLASS = {33E17 (15A52 34A05 34M55 37K10 37N20 60K35 82B41)},
  MRNUMBER = {2066104},
  DOI = {10.1017/S0027763000008801},
  URL = {https://doi.org/10.1017/S0027763000008801}
}

@article{GouldenGuayPaquetNovak2014,
  AUTHOR = {Goulden, I. P. and Guay-Paquet, Mathieu and Novak, Jonathan},
  TITLE = {Monotone {H}urwitz numbers and the {HCIZ} integral},
  JOURNAL = {Ann. Math. Blaise Pascal},
  FJOURNAL = {Annales Math\'ematiques Blaise Pascal},
  VOLUME = {21},
  YEAR = {2014},
  NUMBER = {1},
  PAGES = {71--89},
  ISSN = {1259-1734,2118-7436},
  MRCLASS = {05E10 (33E20)},
  MRNUMBER = {3248222},
  MRREVIEWER = {Shengmao\ Zhu},
  URL = {http://ambp.cedram.org/item?id=AMBP_2014__21_1_71_0}
}

@article{GrelaMajumdarSchehr2021,
  AUTHOR = {Grela, Jacek and Majumdar, Satya N. and Schehr, Gr{\'e}gory},
  TITLE = {Non-intersecting {B}rownian bridges in the flat-to-flat geometry},
  JOURNAL = {J. Stat. Phys.},
  FJOURNAL = {Journal of Statistical Physics},
  VOLUME = {183},
  YEAR = {2021},
  NUMBER = {3},
  PAGES = {Paper No. 49, 35},
  ISSN = {0022-4715,1572-9613},
  MRCLASS = {82C31 (60B20 60G55 60J60 60J65)},
  MRNUMBER = {4273278},
  MRREVIEWER = {Olga\ Izyumtseva},
  DOI = {10.1007/s10955-021-02774-6},
  URL = {https://doi.org/10.1007/s10955-021-02774-6}
}

@article{GuionnetMaida2005,
  AUTHOR = {Guionnet, A. and Ma{\"i}da, M.},
  TITLE = {A {F}ourier view on the {$R$}-transform and related asymptotics of spherical integrals},
  JOURNAL = {J. Funct. Anal.},
  FJOURNAL = {Journal of Functional Analysis},
  VOLUME = {222},
  YEAR = {2005},
  NUMBER = {2},
  PAGES = {435--490},
  ISSN = {0022-1236,1096-0783},
  MRCLASS = {60F10 (46L54)},
  MRNUMBER = {2132396},
  MRREVIEWER = {Sompong\ Dhompongsa},
  DOI = {10.1016/j.jfa.2004.09.015},
  URL = {https://doi.org/10.1016/j.jfa.2004.09.015}
}

@article{HarishChandra1957,
  AUTHOR = {Harish-Chandra},
  TITLE = {Differential Operators on a Semisimple Lie Algebra},
  JOURNAL = {Amer. J. Math.},
  VOLUME = {79},
  NUMBER = {1},
  YEAR = {1957},
  PAGES = {87--120},
  ISSN = {0002-9327,1080-6377},
  URL = {http://www.jstor.org/stable/2372387}
}

@article{ItzyksonZuber1980,
  AUTHOR = {Itzykson, C. and Zuber, J. B.},
  TITLE = {The planar approximation. {II}},
  JOURNAL = {J. Math. Phys.},
  FJOURNAL = {Journal of Mathematical Physics},
  VOLUME = {21},
  YEAR = {1980},
  NUMBER = {3},
  PAGES = {411--421},
  ISSN = {0022-2488,1089-7658},
  MRCLASS = {81E99 (81G05)},
  MRNUMBER = {562985},
  DOI = {10.1063/1.524438},
  URL = {https://doi.org/10.1063/1.524438}
}

@article{JKMS2023,
  AUTHOR = {Jafferis, Daniel Louis and Kolchmeyer, David K. and Mukhametzhanov, Baur and Sonner, Julian},
  TITLE = {Jackiw-{T}eitelboim gravity with matter, generalized eigenstate thermalization hypothesis, and random matrices},
  JOURNAL = {Phys. Rev. D},
  FJOURNAL = {Physical Review D},
  VOLUME = {108},
  YEAR = {2023},
  NUMBER = {6},
  PAGES = {066015},
  ISSN = {2470-0010,2470-0029},
  MRCLASS = {81T35 (81T32 83C80)},
  MRNUMBER = {4654676},
  DOI = {10.1103/PhysRevD.108.066015},
  URL = {https://doi.org/10.1103/PhysRevD.108.066015}
}

@article{KarlinMcGregor1959,
  AUTHOR = {Karlin, Samuel and McGregor, James},
  TITLE = {Coincidence probabilities},
  JOURNAL = {Pacific J. Math.},
  FJOURNAL = {Pacific Journal of Mathematics},
  VOLUME = {9},
  YEAR = {1959},
  PAGES = {1141--1164},
  ISSN = {0030-8730,1945-5844},
  MRCLASS = {60.00},
  MRNUMBER = {0114248},
  MRREVIEWER = {F.\ L.\ Spitzer},
  URL = {http://projecteuclid.org/euclid.pjm/1103038889}
}

@article{KatoriTanemura2003,
  AUTHOR = {Katori, Makoto and Tanemura, Hideki},
  TITLE = {Noncolliding {B}rownian motions and {H}arish-{C}handra formula},
  JOURNAL = {Electron. Commun. Probab.},
  FJOURNAL = {Electronic Communications in Probability},
  VOLUME = {8},
  YEAR = {2003},
  PAGES = {112--121},
  ISSN = {1083-589X},
  MRCLASS = {82B41 (22E30 43A80 60G50 60J65 82B26)},
  MRNUMBER = {2042750},
  DOI = {10.1214/ECP.v8-1076},
  URL = {https://doi.org/10.1214/ECP.v8-1076}
}

@incollection{Kuijlaars2010,
  AUTHOR = {Kuijlaars, Arno B. J.},
  TITLE = {Multiple orthogonal polynomials in random matrix theory},
  BOOKTITLE = {Proceedings of the {I}nternational {C}ongress of {M}athematicians. Volume {III}},
  PAGES = {1417--1432},
  PUBLISHER = {Hindustan Book Agency},
  ADDRESS = {New Delhi},
  YEAR = {2010},
  ISBN = {9789814324359},
  MRCLASS = {60B20 (15B52 42C05)},
  MRNUMBER = {2827849},
  MRREVIEWER = {Steven\ B.\ Damelin}
}

@article{MaillardKrzakalaMezardZdeborova2022,
  AUTHOR = {Maillard, Antoine and Krzakala, Florent and M{\'e}zard, Marc and Zdeborov{\'a}, Lenka},
  TITLE = {Perturbative construction of mean-field equations in extensive-rank matrix factorization and denoising},
  JOURNAL = {J. Stat. Mech. Theory Exp.},
  FJOURNAL = {Journal of Statistical Mechanics: Theory and Experiment},
  YEAR = {2022},
  NUMBER = {8},
  PAGES = {083301},
  ISSN = {1742-5468},
  MRCLASS = {68T05 (15A23 15B52 60B20 62H25 82D30)},
  MRNUMBER = {4504819},
  DOI = {10.1088/1742-5468/ac7e4c},
  URL = {https://doi.org/10.1088/1742-5468/ac7e4c}
}

@book{Mehta2004,
  AUTHOR = {Mehta, Madan Lal},
  TITLE = {Random Matrices},
  SERIES = {Pure and Applied Mathematics (Amsterdam)},
  VOLUME = {142},
  EDITION = {3},
  PUBLISHER = {Elsevier/Academic Press},
  ADDRESS = {Amsterdam},
  YEAR = {2004},
  PAGES = {xviii+688},
  ISBN = {0-12-088409-7},
  MRCLASS = {82-02 (15-02 15A52 60B99 60K35 82B41)},
  MRNUMBER = {2129906}
}

@article{PaccoRos2023,
  AUTHOR = {Pacco, Alessandro and Ros, Valentina},
  TITLE = {Overlaps between eigenvectors of spiked, correlated random matrices: from matrix principal component analysis to random {G}aussian landscapes},
  JOURNAL = {Phys. Rev. E},
  FJOURNAL = {Physical Review E},
  VOLUME = {108},
  YEAR = {2023},
  NUMBER = {2},
  PAGES = {024145},
  ISSN = {2470-0045,2470-0053},
  MRCLASS = {60B20 (62H25)},
  MRNUMBER = {4646713},
  DOI = {10.1103/PhysRevE.108.024145},
  URL = {https://doi.org/10.1103/PhysRevE.108.024145}
}

@inproceedings{TroianiErbaKrzakalaMaillardZdeborova2022,
  AUTHOR = {Troiani, Emanuele and Erba, Vittorio and Krzakala, Florent and Maillard, Antoine and Zdeborov{\'a}, Lenka},
  TITLE = {Optimal Denoising of Rotationally Invariant Rectangular Matrices},
  BOOKTITLE = {Proceedings of the 2nd Mathematical and Scientific Machine Learning Conference},
  SERIES = {Proceedings of Machine Learning Research},
  VOLUME = {145},
  EDITOR = {Bruna, Joan and Hesthaven, Jan and Zdeborov{\'a}, Lenka},
  PAGES = {1--23},
  YEAR = {2022},
  PUBLISHER = {PMLR}
}

@article{ZinnJustin2002,
  AUTHOR = {Zinn-Justin, P.},
  TITLE = {{HCIZ} integral and 2{D} {T}oda lattice hierarchy},
  JOURNAL = {Nuclear Phys. B},
  FJOURNAL = {Nuclear Physics. B},
  VOLUME = {634},
  YEAR = {2002},
  NUMBER = {3},
  PAGES = {417--432},
  ISSN = {0550-3213,1873-1562},
  MRCLASS = {81R12 (37K10 37K60 81T10)},
  MRNUMBER = {1912027},
  MRREVIEWER = {Nenad\ Manojlovi\'c},
  DOI = {10.1016/S0550-3213(02)00374-7},
  URL = {https://doi.org/10.1016/S0550-3213(02)00374-7}
}

@article{AdlerShiotaVanMoerbeke1998,
  AUTHOR = {Adler, M. and Shiota, T. and van Moerbeke, P.},
  TITLE = {Random matrices, {V}irasoro algebras, and noncommutative {KP}},
  JOURNAL = {Duke Math. J.},
  FJOURNAL = {Duke Mathematical Journal},
  VOLUME = {94},
  YEAR = {1998},
  NUMBER = {2},
  PAGES = {379--431},
  ISSN = {0012-7094,1547-7398},
  MRCLASS = {58F07 (17B68 60H25)},
  MRNUMBER = {1638599},
  MRREVIEWER = {Kazuhiro\ Hikami},
  DOI = {10.1215/S0012-7094-98-09417-0},
  URL = {https://doi.org/10.1215/S0012-7094-98-09417-0}
}

@article{AdlerVanMoerbekeVanderstichelen2012,
  AUTHOR = {Adler, Mark and van Moerbeke, Pierre and Vanderstichelen, Didier},
  TITLE = {Non-intersecting {B}rownian motions leaving from and going to several points},
  JOURNAL = {Phys. D},
  FJOURNAL = {Physica D. Nonlinear Phenomena},
  VOLUME = {241},
  YEAR = {2012},
  NUMBER = {5},
  PAGES = {443--460},
  ISSN = {0167-2789,1872-8022},
  MRCLASS = {60J65},
  MRNUMBER = {2878925},
  MRREVIEWER = {Marcel\ Ortgiese},
  DOI = {10.1016/j.physd.2011.10.019},
  URL = {https://doi.org/10.1016/j.physd.2011.10.019}
}

@article{HarnadOrlov2003,
  AUTHOR = {Harnad, J. and Orlov, A. Yu.},
  TITLE = {Scalar products of symmetric functions and matrix integrals},
  JOURNAL = {Teoret. Mat. Fiz.},
  FJOURNAL = {Teoreticheskaya i Matematicheskaya Fizika},
  VOLUME = {137},
  YEAR = {2003},
  NUMBER = {3},
  PAGES = {375--392},
  ISSN = {0564-6162,2305-3135},
  MRCLASS = {82B23 (05E05 33C80 35A30 35Q53 37K10 81R12)},
  MRNUMBER = {2084148},
  MRREVIEWER = {Anjan\ Kundu},
  DOI = {10.1023/B:TAMP.0000007916.13779.17},
  URL = {https://doi.org/10.1023/B:TAMP.0000007916.13779.17}
}

@misc{Orlov2002,
  AUTHOR = {Orlov, A. Yu.},
  TITLE = {Tau Functions and Matrix Integrals},
  YEAR = {2002},
  EPRINT = {math-ph/0210012},
  ARCHIVEPREFIX = {arXiv},
  PRIMARYCLASS = {math-ph}
}
\end{document}